\documentclass[letterpaper, 10 pt, conference]{ieeeconf}   
\IEEEoverridecommandlockouts                              
\usepackage{amsmath}
\usepackage{amssymb} 
\usepackage{graphicx} 
\usepackage{caption}
\usepackage{subcaption}

\usepackage{todonotes}

\usepackage{algorithm}
\usepackage{algpseudocode}

\usepackage[font=small]{caption}

\usepackage{lipsum}

\newtheorem{thm}{Theorem}

\newtheorem{remark}{Remark}

\newtheorem{lemma}{Lemma}
\newtheorem{prob}{Problem}
\newtheorem{hyp}{Assumption}

\DeclareMathOperator{\var}{var}
\DeclareMathOperator{\cov}{cov}
\DeclareMathOperator{\corr}{cor}

\DeclareMathOperator{\sign}{sgn}
\DeclareMathOperator{\diag}{diag}
\newcommand{\bsx}{\boldsymbol{x}}
\newcommand{\bsy}{\boldsymbol{y}}
\newcommand{\bsX}{\boldsymbol{X}}
\newcommand{\bsxi}{\boldsymbol{\xi}}
\newcommand{\bschi}{\boldsymbol{\chi}}
\newcommand{\bsv}{\boldsymbol{v}}
\newcommand{\bsw}{\boldsymbol{w}}
\newcommand{\R}{\mathbb R}
\newcommand{\Real}{\mathrm{Re}}

\title{\LARGE \bf
The correlated variability control problem: a dominant approach*
}

\author{Marcela Ordorica Arango$^{1}$ and Alessio Franci$^{1}$
\thanks{*This work was supported by UNAM-DGAPA PAPIIT grant n. 102420, and Conacyt grant n. CB-A1-S-10610.}
\thanks{$^{1}$Math Department, National Autonomous University of Mexico, 04510, Mexico City, Mexico
        {\tt\small marcela.ordorica@gmail.com}, {\tt\small afranci@ciencias.unam.mx}}
}

\begin{document}

\maketitle
\thispagestyle{empty}
\pagestyle{empty}

\begin{abstract}
Given a population of interconnected input-output agents repeatedly exposed to independent random inputs, we talk of correlated variability when agents' outputs are variable (i.e., they change randomly at each input repetition) but correlated (i.e., they do not vary independently across input repetitions). Correlated variability appears at multiple levels in neuronal systems, from the molecular level of protein expression to the electrical level of neuronal excitability, but its functions and origins are still debated.
Motivated by advancing our understanding of correlated variability, we introduce the (linear) {\it correlated variability control problem} as the problem of controlling steady-state correlations in a linear dynamical network in which agents receive independent random inputs. Although simple, the chosen setting reveals important connections between network structure, in particular, the existence and the dimension of dominant (i.e., slow) dynamics in the network, and the emergence of correlated variability.
\end{abstract}

\section{INTRODUCTION}
Correlated variability is ubiquitous in neuronal systems. Ion channel expression in a population of homogeneous neurons exhibits correlated variability: two neurons of the same type can express very different densities of ion channels but the way in which ion channel density varies across neurons is correlated~\cite{schulz2007quantitative,tran2019ionic}. The role of correlated variability at the molecular level of ion channel expression is debated but it is thought to help finding multiple solutions to the same neural design problem. We recently suggested that correlated variability in ion channel expressions emerges from the dynamical properties of an underlying molecular regulatory networks~\cite{franci2020positive}.

Correlated variability is also observed in the electrical activity of neurons in response to incoming stimuli. When a same stimulus is repeatedly presented to a neuronal population, the intensity of neural response varies across stimulus repetitions but variability in neuronal responses is correlated between neurons~\cite{kohn2016correlations}. Correlated variability is known to shape information coding capabilities of large neuronal populations~\cite{abbott1999effect} but a number of other functions have been explored like the modulation of working memory~\cite{leavitt2017correlated} and the formation of neural assemblies through synaptic plasticity~\cite{averbeck2006neural}. The origins of correlated variability in the electrical activity of neuronal populations is debated~\cite{kohn2016correlations}, but recurrent connections seem to play a fundamental role~\cite{pernice2018interpretation}.

To the best of our knowledge, the problem of controlling correlated variability has never been tackled from a control-theoretical perspective. Here, we give a first step toward addressing this problem in a linear control setting by considering a network of recurrently interconnected scalar agents under the effect of independent random inputs. The main results we prove reveal that the existence of a solution to the correlated variability control problem and the dimension of the solution set are tightly linked to the existence of a dominant (slow) eigenvalue of the network dynamics and, more precisely, to the algebraic and geometric multiplicity of this eigenvalue. Our results provide a first methodology to translate our understanding of correlated variability in biological neural systems into engineered neuromorphic controlled system and artificial neural networks.

\section{NOTATION AND PRELIMINARIES}
Given\footnote{We refer the reader to~\cite{durrett2019probability} for details about measure theory and probability} a probability space $(\Omega, \mathcal{F},P)$, a \textit{random variable} is a measurable function $X:\Omega\to\mathbb{R}$.
A \textit{normally distributed random variable} $X$ is a random variable whose probability density function is given by $f(x)=\frac{1}{\sigma \sqrt{2\pi}}e^{-\frac{1}{2}(\frac{x-\mu}{\sigma})^2}$, where  $\sigma>0$, $\mu\in \mathbb{R}$ is called the \textit{mean} of $X$ and $\sigma^2$ is its \textit{variance}.
Two random variables $X$ and $Y$ are said to be \textit{independent} if, for all Borel sets $C$ and $D$, $P(X\in C,Y\in D)=P(X\in C)P(Y\in D)$.
$\mathcal N_{n}(0,1)$ denotes the space of vectors of $n$ normally distributed independent random variables with zero mean and unitary variance. When it exists, the \textit{expected value} of a random variable $X$ is defined as $\mathbb{E}[X]=\int_{\Omega} XdP$.
The \textit{covariance} between two random variables $X$ and $Y$ is $\cov(X,Y)=\mathbb{E}[XY]-\mathbb{E}[X]\mathbb{E}[Y]$. If $X$ and $Y$ are independent, then $\cov(X,Y)=0$. The covariance function is bilinear, i.e., for $a,b,c,d\in\R$, $\cov(aX_1+bX_2,cX_3+dX_4)=ac\cov(X_1,X_3)+ad\cov(X_1,X_4)+bc\cov(X_2,X_3)+bd\cov(X_2,X_4)$. The \textit{variance} of a random variable $X$ is $\var(X)=\cov(X,X)\geq 0$. The \textit{correlation coefficient} between the random variables $X$ and $Y$ such that $\var(X),\var(Y)>0$ is $\corr(X,Y)=\frac{\cov(X,Y)}{\sqrt{\var(X)\var(Y)}}$. Given a vector of $n$ random variables $\bsX=(X_1,X_2,\ldots,X_n)$, the \textit{covariance matrix} of $\bsX$ is the matrix $\Sigma$ with entries $\Sigma_{ij}=[\cov(X_i,X_j)]$. For $n=2$, the covariance matrix defines a {\it covariance ellipse}, which is the ellipse whose axis are the eigenvectors of $\Sigma$ and whose axis lengths are the square root of the respective eigenvalues.
The symbol $\delta_{ij}$ denotes \textit{Kronecker's delta}: $\delta_{ij}=0$ if $i\neq j$ and $\delta_{ij}=1$ if $i=j$. For any matrix $A\in \mathbb{R}^{n\times n}$ and an eigenvalue $\lambda$ of $A$, $\mu_A(\lambda)$ denotes the \textit{algebraic multiplicity} of $\lambda$ and $\gamma_A(\lambda)\leq \mu_A(\lambda)$ its \textit{geometric multiplicity}. When $\gamma_A(\lambda) < \mu_A(\lambda)$, $A$ does not admit a base of eigenvectors, in which case we resort to generalized eigenvectors and the associated Jordan's canonical form.
An eigenvalue $\lambda_1$ of $A$ is \textit{dominant} if $\lambda_1$ is real and all other eigenvalues of $A$ $\lambda_2,\ldots, \lambda_n$ are such that $\Real(\lambda_i)\leq\lambda_1$. A \textit{dominant eigenvector} is an eigenvector associated to a dominant eigenvalue. A vector is \textit{positive} if all its entries are positive. $\diag(d_1,\ldots,d_n)$ denotes the diagonal matrix with entries $d_1,\ldots,d_n$. A \textit{Metzler matrix} is a matrix such that all its off-diagonal entries are nonnegative. A matrix is \textit{reducible} if it's similar to a matrix of he form 
$\begin{bmatrix}
M_1&M_2\\0&M_3
\end{bmatrix}$. If a matrix is not reducible, it's \textit{irreducible}. Finally, a matrix is \textit{Hurwitz} if all its eigenvalues have strictly negative real part. $\R\{\bsv_1,\ldots,\bsv_m\}$ denotes the subspace spanned spanned by $\bsv_1,\ldots,\bsv_m$. Given two vectors $\bsx,\bsy\in\R^n$, $\langle\bsx,\bsy\rangle=\sum_{i=1}^n x_i,y_i$ denotes the standard Euclidean scalar product between them.

\section{THE CORRELATED VARIABILITY CONTROL PROBLEM AND PRELIMINARY RESULTS}

Consider the following random linear control system
    \begin{equation}
    \label{eq:main_system}
         \dot \bsx = A\bsx + \bsxi\,,
    \end{equation}
where $\bsx=[x_i]_{i=1}^n\in\R^n$ is the state, $A=[a_{ij}]_{i,j=1}^n\in\R^{n\times n}$ is Hurwitz, and $\bsxi=[\xi_i]_{i=1}^n\in \mathcal N_{n}(0,1)$ are random inputs. We interpret $A$ as a weighted signed adjacency matrix, i.e., $a_{ij}$, $i\neq j$, determines the network interaction between variable $x_i$ and variable $x_j$. If $a_{ij}=0$, there are no direct network interactions between $x_i$ and $x_j$, whereas if $a_{ij}>0$ ($<0$) the interaction between the two variables is excitatory (inhibitory). Diagonal terms $a_{ii}$ model internal dynamics of variable $x_i$. The exponentially stable equilibrium $\bschi=[\chi_k]_{k=1}^n$ of model~\eqref{eq:main_system} satisfies $\bschi=-A^{-1}\bsxi$. Thus, $\bschi$ is also a vector of random variables. We talk of a {\it random equilibrium}. 

\vspace{1mm}
\begin{prob}[Correlated variability control]\label{prob:main}
Given $c_{i,j}$, $i,j\in\{1,\ldots,n\}$, $i\neq j$, such that $-1\leq c_{i,j}=c_{j,i}\leq 1$, find a Hurwitz matrix $A$ such that the random equilibrium $\bschi$ satisfies $\corr(\chi_i,\chi_j)=c_{i,j}$.
\end{prob}
\vspace{1mm}

The goal of this paper is to determine necessary conditions on the network structure defined by $A$ such that Problem~\ref{prob:main} admits a solution, at least for some choices of the desired correlations $c_{i,j}$, and to determine the geometry of the solution set, in case it is not empty.\footnote{Of course, Problem~\ref{prob:main} could be solved computationally using brute force. Indeed, $\chi_k=\sum_{l=1}^n [A^{-1}]_{kl}\xi_l$ and therefore $\var(\chi_k)=\sum_{l=1}^n ([A^{-1}]_{kl})^2$ and $\cov(\chi_i,\chi_j)=\sum_{l=1}^n [A^{-1}]_{il}[A^{-1}]_{jl}$, where we used bilinearity of the covariance function and the fact that $\var(\xi_k)=1$ and $\cov(\xi_i,\xi_j)=0$. If follows that
$
\corr(\chi_i,\chi_j)=\frac{\sum_{l=1}^n [A^{-1}]_{il}[A^{-1}]_{jl}}{\sqrt{\sum_{l=1}^n ([A^{-1}]_{il})^2\sum_{l=1}^n ([A^{-1}]_{jl})^2}}
$,
which, in the context of Problem~\ref{prob:main}, leads to an intricate implicit system of equations for the elements of $A$ that might or might not admit a solution. In either case, such a brute force approach is not informative about which network structures, as determined by $A$, lead to a solution for Problem~\ref{prob:main}, or what the geometry of the solution set look like.} In the following two sections, we illustrate two extreme cases leading respectively to null and full (anti)correlations. In both cases, Problem~\ref{prob:main} has no solution except for very specific choices of the desired correlations.

\subsection{Null correlation in the absence of network interactions}

\begin{thm}\label{thm:no corr}
Suppose $A$ is diagonal. Then , for all $i,j\in \{1,\ldots,n\}$, $i\neq j$, $\corr(\chi_i,\chi_j)=0$ and therefore Problem~\ref{prob:main} is unsolvable whenever $c_{i,j}\neq 0$.
\end{thm}
\begin{proof}
Observe that $\chi_i=-\frac{\xi_i}{a_{ii}}$.  By bilinearity of the covariance function, it follows that
$\cov(\chi_i,\chi_j)=\frac{1}{a_{ii}a_{jj}}\cov(\xi_i,\xi_j)=\frac{1}{a_{ii}a_{jj}}\delta_{ij}$.
On the other hand, $\var(\chi_i)=\frac{1}{a_{ii}^2}\var(\xi_i)=\frac{1}{a_{ii}^2}$. Hence,
$\corr(\chi_i,\chi_j)=\frac{\cov(\chi_i,\chi_j)}{\sqrt{\var(\chi_i)\var(\chi_j)}}=\delta_{ij}$
and the result follows.
\end{proof}

Theorem~\ref{thm:no corr} proves the intuitive result that in the absence of network interactions (i.e., $a_{ij}=0$ if $i\neq j$) and in the presence of uncorrelated inputs, the network states are also uncorrelated at equilibrium. Figure~\ref{fig: corr diag A} illustrates this result by simulating one thousand instances of model~\eqref{eq:main_system} with $A=\diag(-0.54,-0.28,-0.14)$. The resulting equilibrium cloud appears as 3-dimensional ellipse whose axis are parallel to the coordinate axis. Projecting equilibria on the three coordinate planes, the measured correlations are (close to, due to finite sample size) zero. 

\begin{figure}
  \begin{subfigure}[b]{.21\textwidth}
    \centering
    \includegraphics[width=\linewidth]{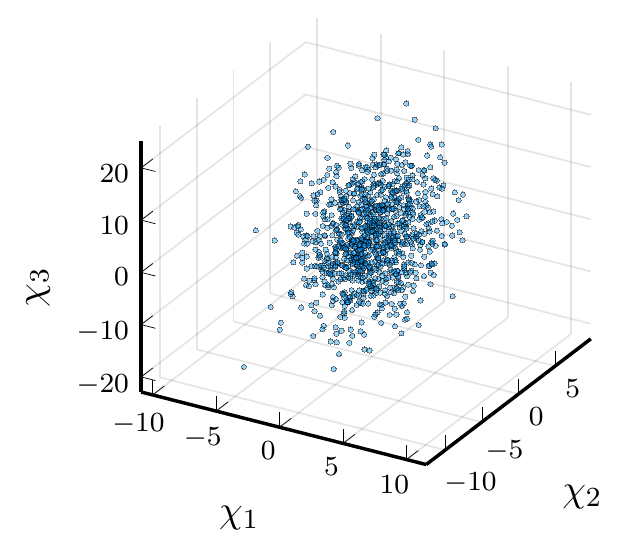}
  \end{subfigure}
  \hfill
  \begin{subfigure}[b]{.21\textwidth}
    \centering
    \includegraphics[width=\linewidth]{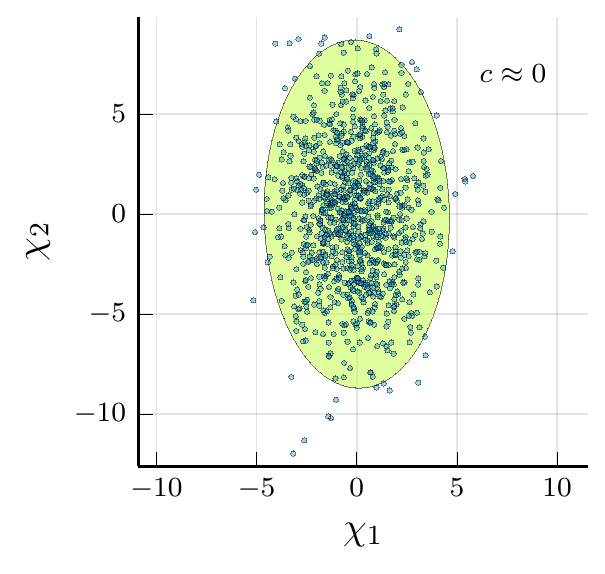}
  \end{subfigure}
  \medskip
  \begin{subfigure}[b]{.21\textwidth}
    \centering
    \includegraphics[width=\linewidth]{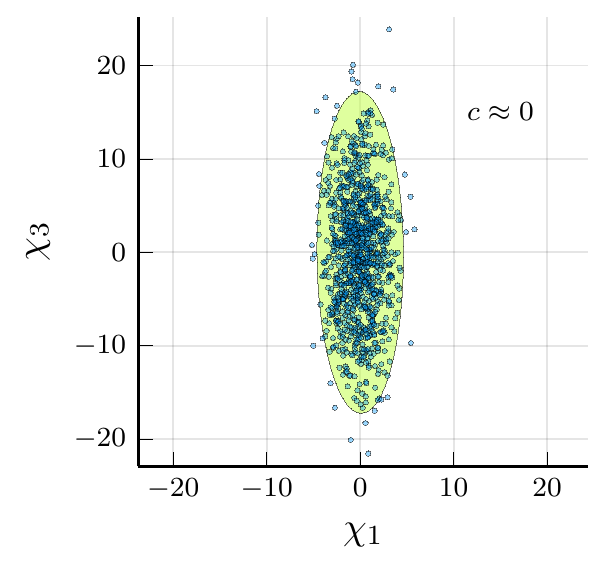}
  \end{subfigure}
  \hfill
  \begin{subfigure}[b]{.21\textwidth}
    \centering
    \includegraphics[width=\linewidth]{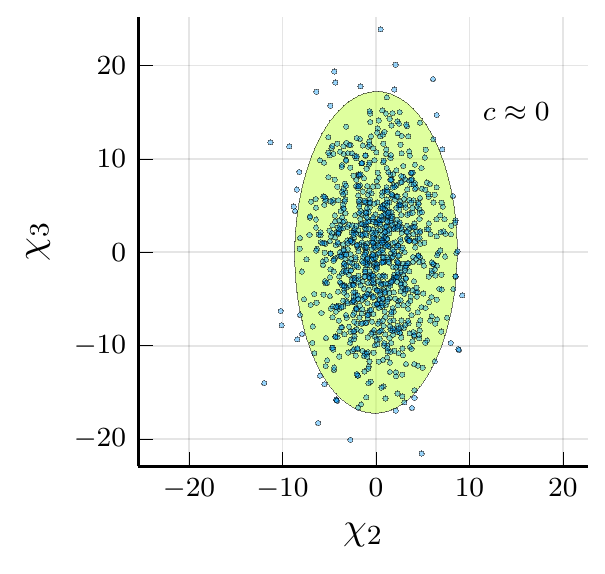}
  \end{subfigure}
\caption{Steady state solution for one thousand instances of model~\eqref{eq:main_system} with $A=\diag(-0.54,-0.28,-0.14)$. The computed correlations are close to zero as predicted. Yellow ellipses are covariance ellipses.
}
\label{fig: corr diag A}
\end{figure}

\subsection{Full (anti)correlations in singular 1-dominant networks}

\begin{thm}\label{thm: full corr}
Suppose $A$ has a positive dominant eigenvector $\bsv_1$ and associated eigenvalue $\lambda_1<0$. Let $\lambda_2,
\ldots,\lambda_n$ be the remaining eigenvalues of $A$ satisfying $\Real(\lambda_i)<\lambda_1$, $i=2,\ldots,n$. Then, in the singular limit $\varepsilon=\frac{-1}{\max_{i>1}\{\Real(\lambda_i)\}}\to0$ and fixed $\lambda_1$, $\corr(\chi_i,\chi_j)=1$ for all $i,j\in\{1,\ldots,n\}$ and therefore Problem~\ref{prob:main} is unsolvable whenever $c_{i,j}\neq 1$.
\end{thm}
\begin{proof}
Let $U$ be the matrix that transforms $A$ in its Jordan canonical form $J$, i.e., $A=UJU^{-1}$, with $J=
\begin{bmatrix}
    \lambda_1 & 0\\
    0 & J_2
\end{bmatrix}$.
Let $\tilde{\bsx}=U^{-1}\bsx$. Then,
\begin{equation}
\label{eq:change_coor}
    \begin{aligned}
        \dot{\tilde{x}}_1&=\lambda_1\tilde{x_1}+\langle U^{-1}_{1 \cdot}, [\xi_j]_{j=1}^n \rangle,\\
        [\dot{\tilde{x}}_i]_{i=2}^n&=J_2[x_j]_{j=1}^n+[\langle U_{i\cdot}^{-1}, [\xi_j]_{j=1}^n\rangle]_{i=1}^n\,.
    \end{aligned}
\end{equation}
It follows along the same lines as~\cite[Theorem~1]{franci2020positive} that in the limit $\varepsilon\to0$, model~\eqref{eq:change_coor} reduces to the slow dynamics
\begin{equation*}
    \begin{aligned}
        \dot{\tilde{x}}_1&=\lambda_1\tilde{x_1}+\langle U^{-1}_{1 \cdot}, [\xi_j]_{j=1}^n \rangle,\\
        \tilde x_i&=0,\quad i=2,\ldots,n\,.
    \end{aligned}
\end{equation*}
Let $\tilde\chi_1=-\frac{\langle U^{-1}_{1 \cdot}, [\xi_j]_{j=1}^n \rangle}{\lambda_1}$. Then, in the limit $\varepsilon\to 0$, $\chi_i=U_{i1}\tilde{\chi}_1$ and therefore, at equibrilium, $\corr(\chi_i,\chi_j)=\corr(U_{i1}\tilde\chi_1,U_{j1}\tilde\chi_1)=\frac{\cov(U_{i1}\tilde\chi_1,U_{j1}\tilde\chi_1)}{\sqrt{\var(U_{i1}\tilde\chi_1)\var(U_{j1}\tilde\chi_1)}}=\frac{U_{i1}U_{j1}\cov(\tilde\chi_1,\tilde\chi_1)}{\sqrt{U_{i1}^2U_{j1}^2\var(\tilde\chi_1)^2}}=1$.
\end{proof}

Theorem~\ref{thm: full corr} shows that if $A$ has a positive dominant eigenvector with a negative associated eigenvalue $\lambda_1$ (e.g., $A$ is Metzler and irreducible~\cite[Lemma VIII.1]{angeli2003monotone} or $A$ is eventually positive~\cite[Theorem 5]{giordano2017interaction}), and if the separation between the slow eigenvalue and the rest of the spectrum is sufficiently large, then all pairs of variables are fully correlated. Figure~\ref{fig: corr 1dom} illustrates this result by simulating one thousand instances of model~\eqref{eq:main_system} with $A$ constructed through the inverse of the unitary change of base that diagonalizes it to possess the following (eigenvalue,eigenvector) pairs\footnote{Note that three chosen eigenvectors are orthonormal. For clarity, vectors' entries are reported up to the second digit.}: $(-0.01,(0.45,0.81,0.36))$, $(-0.2,(-0.81,0.54,-0.2))$, $(-1.0,(-0.36,-0.2,0.9))$. Observe that the resulting equilibrium cloud appears as a three-dimensional ellipse sharply elongated along the dominant (slow) direction because the input-to-state gain along this direction is much larger than along non-dominant ones.
Projecting equilibria on the three coordinate planes, the measured correlations are all (close to, due to finite sample size and finite $\varepsilon=0.01$) one.

\begin{figure}
  \begin{subfigure}[b]{.21\textwidth}
    \centering
    \includegraphics[width=\linewidth]{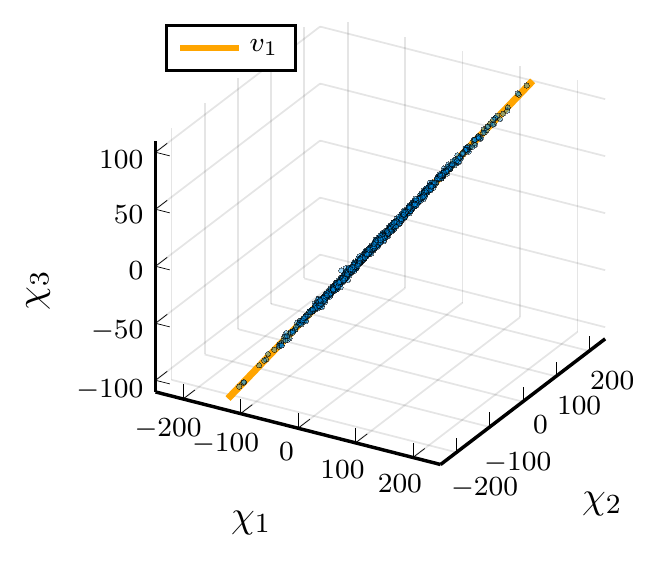}
  \end{subfigure}
  \hfill
  \begin{subfigure}[b]{.21\textwidth}
    \centering
    \includegraphics[width=\linewidth]{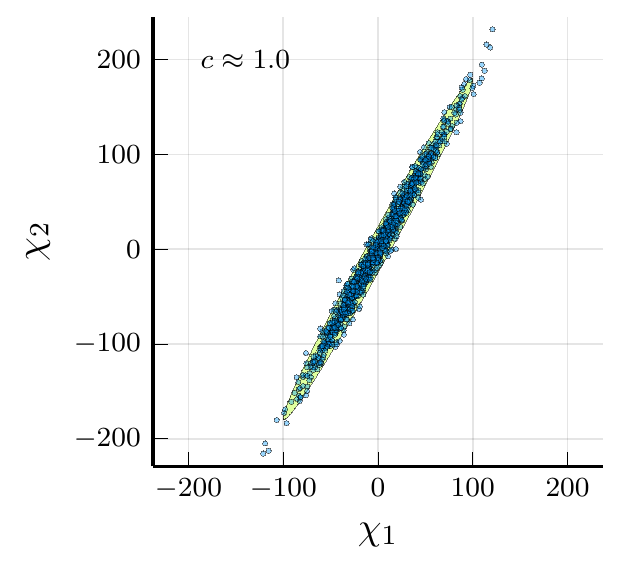}
  \end{subfigure}
  \medskip
  \begin{subfigure}[b]{.21\textwidth}
    \centering
    \includegraphics[width=\linewidth]{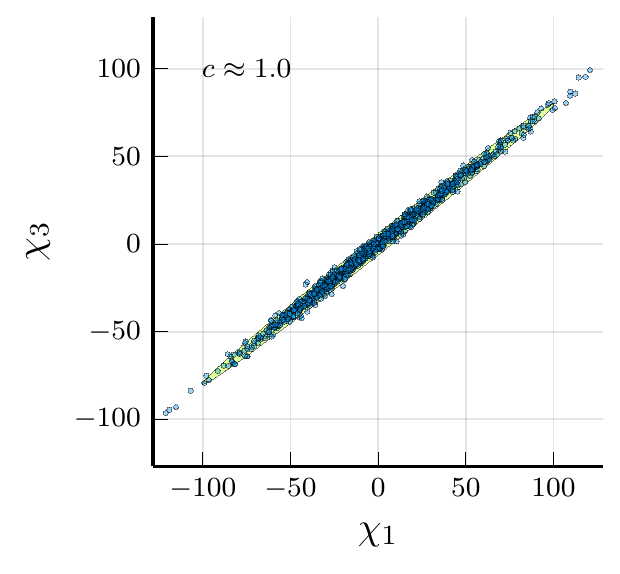}
  \end{subfigure}
  \hfill
  \begin{subfigure}[b]{.21\textwidth}
    \centering
    \includegraphics[width=\linewidth]{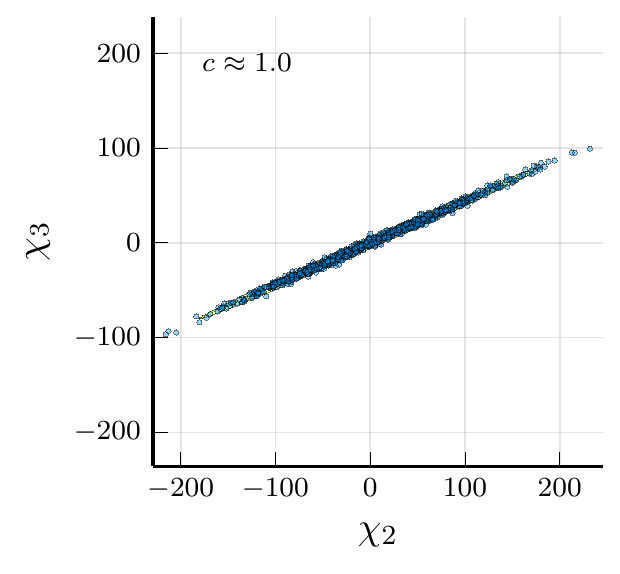}
  \end{subfigure}
\caption{Steady state solution for one thousand instances of model~\eqref{eq:main_system} in case $A$ possesses a single dominant direction. The computed correlations are close to one as predicted. Yellow ellipses are covariance ellipses.}
\label{fig: corr 1dom}
\end{figure}

\begin{remark}
For $\epsilon>0$ and sufficiently small, if follows along the same line as~\cite[Proposition~5]{franci2020positive}, that, $\corr(x_i^*,x_j^*)=1-\mathcal O(\varepsilon)$ for all $i,j\in\{1,\ldots,n\}$.
\end{remark}

The following theorem addresses the case in which the dominant eigendirection $\bsv_1$ has mixed-sign entries. In this case, any pair of variables is either fully correlated or fully anticorrelated.

\begin{thm}
Suppose $A$ has a dominant eigendirection $\bsv_1$ with associated eigenvalue $\lambda_1<0$. Let $\lambda_2,
\ldots,\lambda_n$ be the remaining eigenvalues of $A$ satisfying $\Real(\lambda_i)\ll\lambda_1$, $i=2,\ldots,n$. Then, in the singular limit $\varepsilon=\frac{-1}{\max_{i>1}\{\Real(\lambda_i)\}}\to0$ and fixed $\lambda_1$, $\corr(\chi_i,\chi_j)=\sign([\bsv_1]_i[\bsv_1]_j)$ for all $i,j\in\{1,\ldots,n\}$ and therefore Problem~\ref{prob:main} is unsolvable whenever $c_{i,j}\neq \pm 1$.
\end{thm}
\begin{proof}
Follows along the same lines as the proof of Theorem~\ref{thm: full corr} and observing that $U_{i1}=[\bsv_1]_i$.
\end{proof}

\section{CORRELATED VARIABILITY CONTROL IN THE PRESENCE OF A REPEATED DOMINANT EIGENVALUE}

The results in the Theorems~\ref{thm:no corr} and~\ref{thm: full corr} shows that it should {\it a priori} be possible to span the whole range of correlation degrees, from null (no network) to full (strongly 1-dominant network), by suitably changing the network structure. The rationale we follow here is that increasing the dimension of the dominant subspace and designing suitable slow dynamics on it leads a constructive geometric way to solve Problem~\ref{prob:main}.

\subsection{The dominant eigenvalue has algebraic and geometric multiplicity two}

\begin{hyp}\label{Hyp:2 dominance}
$A$ has a real repeated dominant eigenvalue $\lambda_1<0$, with $\gamma_A(\lambda_1)=\mu_A(\lambda_1)=2$, and dominant eigenvectors $\bsv_1$ and $\bsv_2$. Let $\lambda_3,\ldots,\lambda_{n}$, $\Real(\lambda_i)\ll\lambda_1$, $i=3,\ldots,n$ be the remaining eigenvalues of $A$.
\end{hyp}

\begin{lemma}\label{lem:orthonormal NoJor}
Without loss of generality, $\bsv_1$ and $\bsv_2$ can be taken to be orthonormal.
\end{lemma}
\begin{proof}
Any vector in $\R\{\bsv_1,\bsv_2\}$ is also an eigenvector of $A$ with eigenvalue $\lambda_1$. If not already orthonormal, redefine $\bsv_1$ and $\bsv_2$ to be an orthonormal basis of $\R\{\bsv_1,\bsv_2\}$.
\end{proof}

\begin{lemma}\label{lem: unitary No Jor}
There exists is a unitary matrix $Q$ such that $Q^{-1}AQ=\begin{bmatrix}
        T_{\lambda_1}&C\\
        0&B
        \end{bmatrix}$,
where $T_{\lambda_1}=\begin{bmatrix}
    \lambda_1&0\\
    0&\lambda_1
    \end{bmatrix}$,
$C\in \mathbb{R}^{2,n-2}$, $B\in \mathbb{R}^{n-2,n-2}$, and the eigenvalues of $B$ are exactly $\lambda_3, \ldots, \lambda_{n}$.
\end{lemma}
\begin{proof}
Let $\bsw_3,\bsw_2,\ldots,\bsw_n$ be vectors such that $\{\bsv_1,\bsv_2,\bsw_3,\ldots,\bsw_n\}$ is an orthonormal basis of $\mathbb{R}^n$ and take $Q$ to be the matrix whose columns are these vectors. Then, $Q$ is unitary. Furthermore, $Q^{-1}AQ=Q^{-1}
    \begin{bmatrix}
    \lambda_1\bsv_1&\lambda_1\bsv_2&A\bsw_3&\cdots&A\bsw_n\\
    \end{bmatrix}=\left[
    \begin{array}{c|c}
    \begin{matrix}
        \lambda_1&0\\
        0&\lambda_1
        \end{matrix} & C \\
    \hline 
    0 & B
\end{array}
    \right]=\begin{bmatrix}
        T_{\lambda_1}&C\\
        0&B
        \end{bmatrix}$.
To see that the eigenvalues of $B$ are $\lambda_3,\ldots,\lambda_{n}$, notice that $A$ and $Q^{-1}AQ$ are similar and therefore have the same eigenvalues.
\end{proof}

Let $\tilde{x}=Q^{-1}x$. Then,
\begin{subequations}\label{eq: 2 domtilde sys}
\begin{align}
    \dot{\tilde{x}}_1&=\lambda_1\tilde{x}_1+\langle C_{1\cdot},[\tilde{x}_j]_{j=3}^n\rangle+\langle Q_{1\cdot}^{-1},[\xi_j]_{j=1}^n\rangle\\
    \dot{\tilde{x}}_2&=\lambda_1\tilde{x}_2+\langle C_{2\cdot},[\tilde{x}_j]_{j=3}^n\rangle+\langle Q_{2\cdot}^{-1},[\xi_j]_{j=1}^n\rangle\\
    [\dot{\tilde{x}}_i]_{i=3}^n&=B[\tilde{x}_j]_{j=3}^n+[\langle Q_{i\cdot}^{-1},[\xi_j]_{j=1}^n\rangle]_{i=3}^n.
\end{align}
\end{subequations}

\begin{thm}\label{thm: noJor slow dyn}
Under Assumption~\ref{Hyp:2 dominance}, the slow dynamics of model~\eqref{eq:main_system} associated to the singular limit $\varepsilon=\frac{-1}{\max_{i>2}\{\Real(\lambda_i)\}}\to0$ and fixed $\lambda_1$ of model~\eqref{eq: 2 domtilde sys} reads
\begin{subequations}\label{eq:reduced22}
\begin{align}
    \dot{\tilde{x}}_1&=\lambda_1\tilde{x}_1+\langle Q_{1\cdot}^{-1},[\xi_j]_{j=1}^n\rangle\\
    \dot{\tilde{x}}_2&=\lambda_1\tilde{x}_2+\langle Q_{2\cdot}^{-1},[\xi_j]_{j=1}^n\rangle.\\
    \tilde x_i&=0,\quad i=3,\ldots,n\,.
\end{align}
\end{subequations}
Furthermore, the associated critical manifold $\mathcal M=\{\tilde x_3=\cdots=\tilde x_n=0\}$ is exponentially attractive.
\end{thm}
\begin{proof}
Let $B_1=\varepsilon B$. Then $B_1$ has spectrum $\varepsilon\lambda_3,\ldots,\varepsilon\lambda_n$, with $\varepsilon\max_{i>2}\{\Real(\lambda_i)\}=-1$. Then~\eqref{eq: 2 domtilde sys} becomes
\begin{align}\label{eq:reduced22 temp}
    \dot{\tilde{x}}_1&=\lambda_1\tilde{x}_1+\langle C_{1\cdot},[\tilde{x}_j]_{j=3}^n\rangle+\langle Q_{1\cdot}^{-1},[\xi_j]_{j=1}^n\rangle\\
    \dot{\tilde{x}}_2&=\lambda_1\tilde{x}_2+\langle C_{2\cdot},[\tilde{x}_j]_{j=3}^n\rangle+\langle Q_{2\cdot}^{-1},[\xi_j]_{j=1}^n\rangle\\
    \varepsilon[\dot{\tilde{x}}_i]_{i=3}^n&=B_1[\tilde{x}_j]_{j=3}^n+\varepsilon[\langle Q_{i\cdot}^{-1},[\xi_j]_{j=1}^n\rangle]_{i=3}^n.
\end{align}
Because $B_1$ is Hurwitz and its spectrum is bounded away from zero for all $\varepsilon>0$, in the limit $\varepsilon\to0$, model~\eqref{eq:reduced22 temp} reduces to the slow dynamics~\eqref{eq:reduced22} defined on the critical manifold $\mathcal M$, which is also exponentially attractive.
\end{proof}

Let $\tilde\bschi=U^{-1}\chi$.  We are now in condition to compute steady-state correlations of model~\eqref{eq:main_system}.

\begin{thm}
\label{thm:var_cor_tildes}
Under Assumption~\ref{Hyp:2 dominance} and in the singular limit $\varepsilon=\frac{-1}{\max_{i>2}\{\Real(\lambda_i)\}}\to0$ and fixed $\lambda_1$, equilibria of model~\eqref{eq:main_system} satisfy $\var(\tilde\chi_1)=\var(\tilde\chi_2)=\frac{1}{\lambda_1^2},\ \cov(\tilde\chi_1,\tilde\chi_2)=0$.
Furthermore, for all $i,j\in\{1,\ldots,n\}$,
\begin{equation}\label{eq: corr formulas noJor}
   \corr(\chi_i,\chi_j)=\frac{[\bsv_1]_i[\bsv_1]_j+[\bsv_2]_i[\bsv_2]_j}{\sqrt{([\bsv_1]_i^2+[\bsv_2]_i^2)([\bsv_1]_j^2+[\bsv_2]_j^2)}}. 
\end{equation}
\end{thm}
\begin{proof}
 Observe that $\tilde\chi_1=-\frac{\langle Q_{1\cdot}^{-1},[\xi_j]_{j=1}^n\rangle}{\lambda_1,},\ \tilde\chi_2=-\frac{\langle Q_{2\cdot}^{-1},[\xi_j]_{j=1}^n\rangle}{\lambda_1}$.
 Using the facts that $\cov(\xi_i,\xi_j)=\delta_{ij}$ and that the covariance function is bilinear, we have
 $
\cov(\tilde\chi_j,\tilde\chi_l)=\frac{1}{\lambda_1^2}\sum_{k=1}^n Q_{jk}^{-1}Q_{lk}^{-1},\ j,l=1,2
 $.
 The formulas for $\var(\tilde\xi_j)$ and $\cov(\tilde\chi_j,\tilde\chi_l)$, $j,l=1,2$, then follow by recalling that, by Lemma~\ref{lem: unitary No Jor}, $Q_{j\cdot}^{-1}=\bsv_j$, $j=1,2$, and that, by Lemma~\ref{lem:orthonormal NoJor}, $\bsv_1$ and $\bsv_2$ are orthonormal. The formulas for $\corr(\chi_i,\chi_j)$ follow by invoking Theorem~\ref{thm: noJor slow dyn}, which implies that in the limit $\varepsilon\to0$, $\chi_i=[\bsv_1]_i\tilde\chi_1+[\bsv_2]_i\tilde\chi_2$.
\end{proof}

Theorem~\ref{thm:var_cor_tildes} shows that, under Assumption~\ref{Hyp:2 dominance} and in the singular limit of strong dominance $\varepsilon\to0$, the $n$-dimensional covariance ellipse generated by the steady states of model~\eqref{eq:main_system} reduces to a (two-dimensional) circle of radius $|\lambda_1|^{-1}$ on the dominant subspace spanned by $\bsv_1$ and $\bsv_2$. Furthermore it provides explicit formulas~\eqref{eq: corr formulas noJor} in terms of the components of the dominant eigenvectors $\bsv_1$ and $\bsv_2$ for the steady-state correlations $\corr(\chi_i,\chi_j)$.

Expressions~\eqref{eq: corr formulas noJor} can be plugged into out-of-the-box optimization software to find solutions to Problem~\ref{prob:main} but still provide no guarantees about the existence of such solutions nor about the geometry of the possible solution set. We have the following theorem.

\begin{thm}\label{thm: noJor}
Let Assumption~\ref{Hyp:2 dominance} be satisfied and let $n=3$. Then, in the limit $\varepsilon=\frac{-1}{\max_{i>1}\{\Real(\lambda_i)\}}\to0$ and fixed $\lambda_1$, Problem~\ref{prob:main} has solution on an open set of desired correlations $c_{1,2},c_{2,3},c_{1,3}$. Furthermore, when they exist, solutions are isolated.
\end{thm}
\begin{proof}
 Observe that because $\bsv_1$ and $\bsv_2$ are orthonormal, for $n=3$ they are parameterized by three angles $\theta_1,\theta_2,\theta_3$ (e.g., the two angles defining the orientation of $\bsv_1$ and the angle of $\bsv_2$ on the plan orthogonal to $\bsv_1$). Under this parameterization,~\eqref{eq: corr formulas noJor} defines a smooth map
 \begin{align*}
     \Phi:\R^3&\to\R^3\\
     (\theta_1,\theta_2,\theta_3)&\mapsto(\corr(\chi_1,\chi_2),\corr(\chi_2,\chi_3),\corr(\chi_3,\chi_1)).
 \end{align*}
 Furthermore, observe that~$\Phi$ has range $[-1,1]^3$ and therefore, for $n=3$, there exists at least one combination of desired correlations for which Problem~\ref{prob:main} has solution. It is lengthy but straightforward to show that the Jacobian of $\Phi$ is non-singular almost everywhere and therefore by the Open Mapping theorem~\cite{bartle1976} its image is open. Thus, Problem~\ref{prob:main} has a solution on an open set of desired correlations, i.e., the image of $\Phi$. Let $c_{1,2},c_{2,3},c_{1,3}$ be desired correlations for which a solution to Problem~\ref{prob:main} exist. Let $\Phi_{1,2}(\theta_1,\theta_2,\theta_3)=\corr(\chi_1,\chi_2)$ be the first component of $\Phi$. By the Implicit Function theorem applied to $\Phi_{1,2}$, there exists a two-dimensional almost-everywhere (i.e., except at possible singularities) smooth manifold $\mathcal M_{1,2}\subset\R^3$ such that if $\Phi_{1,2}(\theta_1,\theta_2,\theta_3)=c_{1,2}$, then $(\theta_1,\theta_2,\theta_3)\in\mathcal M_{1,2}$. Through the same argument, it follows that simulataneously imposing $\corr(\chi_1,\chi_2)=c_{1,2}$, $\corr(\chi_2,\chi_3)=c_{2,3}$, and $\corr(\chi_3,\chi_1)=c_{3,1}$ implies $(\theta_1,\theta_2,\theta_3)\in\mathcal M_{1,2}\cap \mathcal M_{2,3} \cap \mathcal M_{3,1}$, where $\mathcal M_{2,3}, \mathcal M_{3,1}$ are also two-dimensional almost-everywhere smooth manifolds. Recalling that the intersection of three two-dimensional manifold in $R^3$ is generically zero-dimensional, i.e, made of isolated points, the result follows.
\end{proof}

\begin{figure}
  \begin{subfigure}[b]{.21\textwidth}
    \centering
    \includegraphics[width=\linewidth]{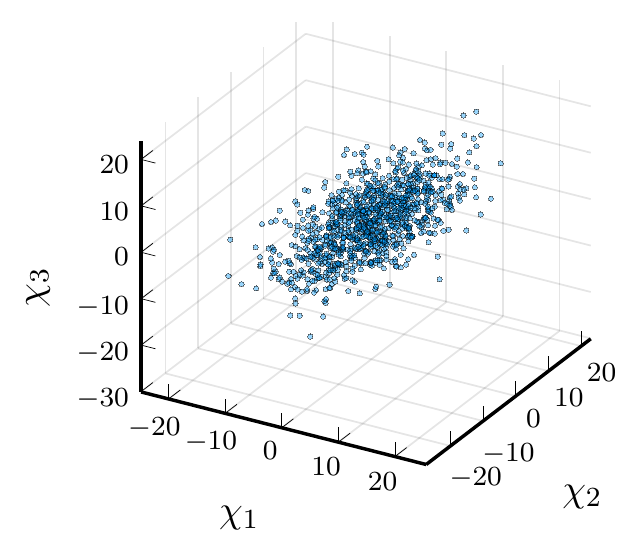}
  \end{subfigure}
  \hfill
  \begin{subfigure}[b]{.21\textwidth}
    \centering
    \includegraphics[width=\linewidth]{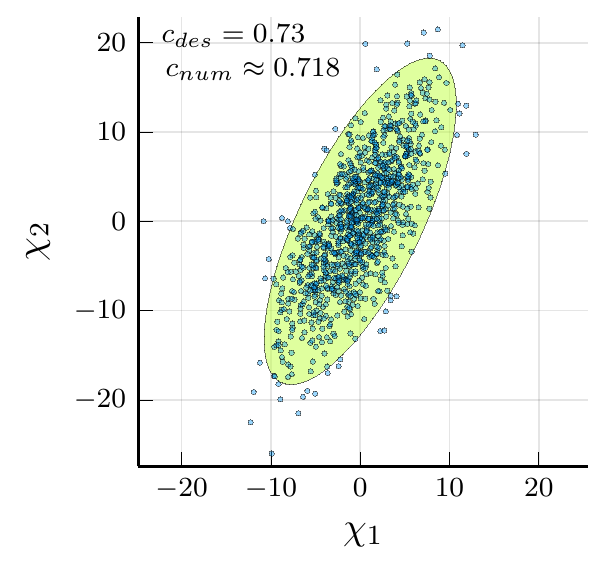}
  \end{subfigure}
  \medskip
  \begin{subfigure}[b]{.21\textwidth}
    \centering
    \includegraphics[width=\linewidth]{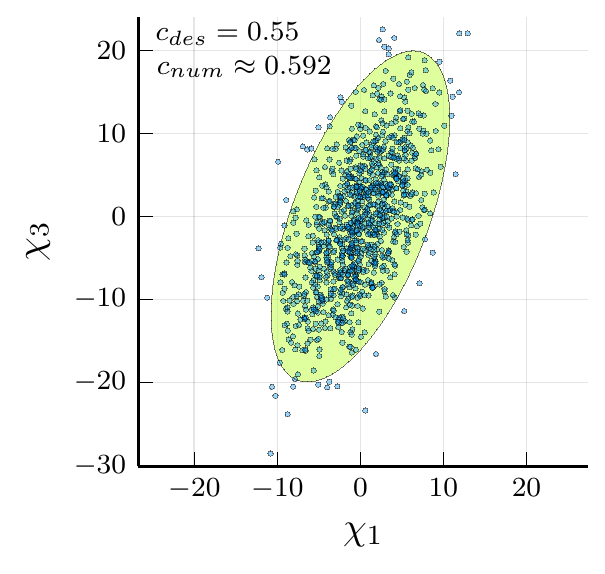}
  \end{subfigure}
  \hfill
  \begin{subfigure}[b]{.21\textwidth}
    \centering
    \includegraphics[width=\linewidth]{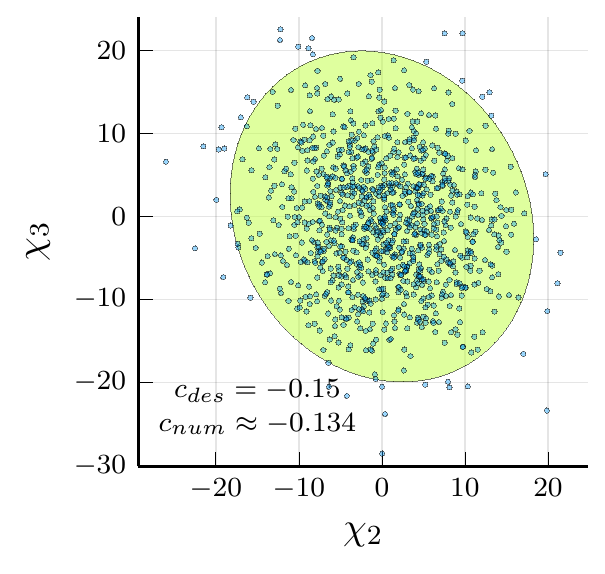}
  \end{subfigure}
\caption{Steady state solution for one thousand instances of model~\eqref{eq:main_system} in case the dominant eigenvalue of $A$ has algebraic and geometric multiplicity two. The computed correlations are close to the desired control values. Yellow ellipses are covariance ellipses.}
\label{fig: corr 2domNJ}
\end{figure}

Figure~\ref{fig: corr 2domNJ} illustrates Theorem~\ref{thm: noJor}. An orthonormal solution $\bsv_1,\bsv_2$ to Problem~\ref{prob:main} was found by plugging expression~\eqref{eq: corr formulas noJor} into the Julia Optimization package {\tt Optim.jl}~\cite{mogensen2018optim}. Given this solution, we assigned $\lambda_1=-0.01$ and $\lambda_3=-1.0$. The eigenvector $\bsv_3$ associated to $\lambda_3$ was taken to be orthonormal to $\bsv_1,\bsv_2$. The Hurwitz matrix $A$ was then constructed through the inverse of the unitary change of base associated to the base $\bsv_1,\bsv_2,\bsv_3$. We then simulated one-thousand instances of model~\eqref{eq:main_system}. As shown in Figure~\ref{fig: corr 2domNJ} numerically computed correlations are close to the associated desired control values. To verify that the solution set is indeed made of isolated points we rerun our algorithm for the same desired correlation values as Figure~\ref{fig: corr 2domNJ} but let the optimization procedure converge to a new solution for five hundred initial conditions in a small neighborhood of the solution corresponding to Figure~\ref{fig: corr 2domNJ}. The result of this experiment is reproduced in Figure~\ref{fig:nubes} (yellow plots). The optimizer consistently converges back to the original solution, which, as predicted by Theorem~\ref{thm: noJor}, shows that no other solutions exist close to it.

\subsection{The dominant eigenvalue has algebraic multiplicity two but geometric multiplicity one}

We develop this section under the following assumption.
\begin{hyp}\label{Hyp:2 dominance Jor}
$A$ has a real repeated dominant eigenvalue $\lambda_1<0$, with $\gamma_A(\lambda_1)=1$, $\mu_A(\lambda_1)=2$, dominant normalized eigenvector $\bsv_1$ and generalized normalized eigenvector $\bsv_2$. Let $\lambda_3,\ldots,\lambda_{n}$, $\Real(\lambda_i)\ll\lambda_1$, $i=3,\ldots,n$ be the remaining eigenvalues of $A$.
\end{hyp}

Let $U$, with $U_{\cdot1}=\bsv_1$ and $U_{\cdot2}=\bsv_2$, be the matrix that transforms $A$ in its Jordan canonical form $J$, i.e., $A=UJU^{-1}$, with $J=\begin{bmatrix}
        J_{\lambda_1}&0\\
        0&J_2
        \end{bmatrix}$,
where $J_{\lambda_1}$ is the two-dimensional Jordan block associated to $\lambda_1$ and $J_2$ contains Jordan blocks $\lambda_3,\ldots,\lambda_n$. Let $\tilde x=U^{-1}x$. Then
\begin{equation}
\label{eq:change_coor_21}
\begin{split}
    \dot{\tilde{x}}_1&=\lambda_1\tilde{x}_1+\tilde{x}_2+\langle U_{1\cdot}^{-1},[\xi_j]_{j=1}^n\rangle,\\
    \dot{\tilde{x}}_2&=\lambda_1\tilde{x}_2+\langle U_{2\cdot}^{-1},[\xi_j]_{j=1}^n\rangle,\\
    [\dot{\tilde{x}}_i]_{i=3}^n&=J_2[\tilde{x}_j]_{j=3}^n+[\langle U_{i\cdot}^{-1},[\xi_j]_{j=1}^n\rangle]_{i=3}^n.
\end{split}
\end{equation}

The proof of the following theorem follows along the same line as the proof of Theorem~\ref{thm: noJor slow dyn}.

\begin{thm}\label{thm: Jor slow dyn}
Under Assumption~\ref{Hyp:2 dominance Jor}, the slow dynamics of model~\eqref{eq:main_system} associated to the singular limit $\varepsilon=\frac{-1}{\max_{i>2}\{\Real(\lambda_i)\}}\to0$ and fixed $\lambda_1$ of model~\eqref{eq: 2 domtilde sys} reads
\begin{subequations}\label{eq:reduced22 jor}
\begin{align*}
    \dot{\tilde{x}}_1&=\lambda_1\tilde{x}_1+\tilde{x}_2+\langle U_{1\cdot}^{-1},[\xi_j]_{j=1}^n\rangle,\\
    \dot{\tilde{x}}_2&=\lambda_1\tilde{x}_2+\langle U_{2\cdot}^{-1},[\xi_j]_{j=1}^n\rangle,\\
    {\tilde{x}}_i&=0,\quad i=1,\ldots,n\,.
\end{align*}
\end{subequations}
Furthermore, the associated critical manifold $\mathcal M=\{\tilde x_3=\cdots=\tilde x_n=0\}$ is exponentially attractive.
\end{thm}

Let $\tilde\bschi=U^{-1}\chi$. We are now in condition to compute steady-state correlations of model~\eqref{eq:main_system}.

\begin{thm}
\label{thm:var_cor_tildes Jor}
Under Assumption~\ref{Hyp:2 dominance Jor} and in the singular limit $\varepsilon=\frac{-1}{\max_{i>2}\{\Real(\lambda_i)\}}\to0$ and fixed $\lambda_1$, equilibria of model~\eqref{eq:main_system} satisfy
\begin{align*}
        &\var({\tilde{\chi}_1})\!=\!\sum_{i=1}^n\!\bigg(\!\frac{-\lambda_1U_{1i}^{-1}\!+\!U_{2i}^{-1}}{\lambda_1^2}\bigg)^2\!\!\!,\var({\tilde{\chi}_2})\!=\!\sum_{i=1}^n\!\bigg(\!\!-\frac{U_{2i}^{-1}}{\lambda_1}\bigg)^2\\
        &\cov({\tilde{\chi}_1},{\tilde{\chi}_2})=-\frac{1}{\lambda_1^3}\sum_{i=1}^n(-\lambda_1U_{1i}^{-1}+U_{2i}^-1)(U_{2i}^{-1}).
\end{align*}
Furthermore, for all $i,j\in\{1,\ldots,n\}$,
\begin{equation}\label{eq: corr formulas Jor}
    \corr(\chi_i,\chi_j)=\frac{\beta_1}{\sqrt{\beta_2\cdot \beta_3}}
\end{equation}
where $\beta_1=[\bsv_1]_i[\bsv_1]_j\var(\tilde{\chi}_1)+([\bsv_1]_i[\bsv_2]_j+[\bsv_2]_i[\bsv_1]_j)\cov(\tilde{\chi}_1,\tilde{\chi}_2)+[\bsv_2]_i[\bsv_2]_j\var(\tilde{\chi}_2)$, $\beta_2=[\bsv_1]_i^2\var(\tilde{\chi}_1)+[\bsv_2]_i^2\var(\tilde{\chi}_2)+2[\bsv_1]_i[\bsv_2]_i\cov(\tilde{\chi}_1,\tilde{\chi}_2)$, $\beta_3=[\bsv_1]_j^2\var(\tilde{\chi}_1)+[\bsv_2]_j^2\var(\tilde{\chi}_2)+2[\bsv_1]_j[\bsv_2]_j\cov(\tilde{\chi}_1,\tilde{\chi}_2)$.
     \end{thm}
\begin{proof}
 The first part of the statement follows by solving for the steady-states of the slow dynamics~\eqref{eq:reduced22 jor} and using properties of the covariance function along the same lines as the proof of Theorem~\ref{thm:var_cor_tildes}. The formulas for $\corr(\chi_i,\chi_j)$ follow by invoking Theorem~\ref{thm: Jor slow dyn}, which implies that in the limit $\varepsilon\to0$, $\chi_i=[\bsv_1]_i\tilde\chi_1+[\bsv_2]_i\tilde\chi_2$.
\end{proof}

To study the geometry of the solution set of Problem~\ref{prob:main} under Assumption~\ref{Hyp:2 dominance Jor} and for $n=3$, we follow the same steps as for Theorem~\ref{thm: noJor}. The key point is observing that in solving $\corr(\chi_i,\chi_j)=c_{i,j}$, with $\corr(\chi_i,\chi_j)$ defined by~\eqref{eq: corr formulas Jor}, for $n=3$ there are five independent variables, that is, the two solid angles $((\theta_1,\theta_2),(\theta_3,\theta_4))$, which define the two normalized vectors $\bsv_1,\bsv_2$, and the slow eigenvalue $\lambda_1$. Hence, imposing values for the three correlations leads to a solution set given by the intersection of three four-dimensional manifolds in $\R^5$, i.e., in general, a two-dimensional manifold.

\begin{thm}\label{thm: Jor}
Let Assumption~\ref{Hyp:2 dominance Jor} be satisfied and let $n=3$. Then, in the limit $\varepsilon=\frac{-1}{\max_{i>1}\{\Real(\lambda_i)\}}\to0$ and fixed $\lambda_1$, Problem~\ref{prob:main} has solution on an open set of desired correlations $c_{1,2},c_{2,3},c_{1,3}$. Furthermore, when they exist, solutions lies on a two-dimensional manifold.
\end{thm}

\begin{figure}
  \begin{subfigure}[b]{.21\textwidth}
    \centering
    \includegraphics[width=\linewidth]{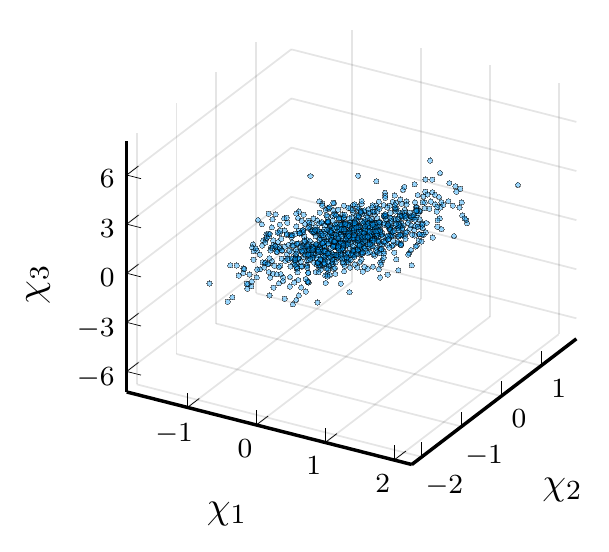}
  \end{subfigure}
  \hfill
  \begin{subfigure}[b]{.21\textwidth}
    \centering
    \includegraphics[width=\linewidth]{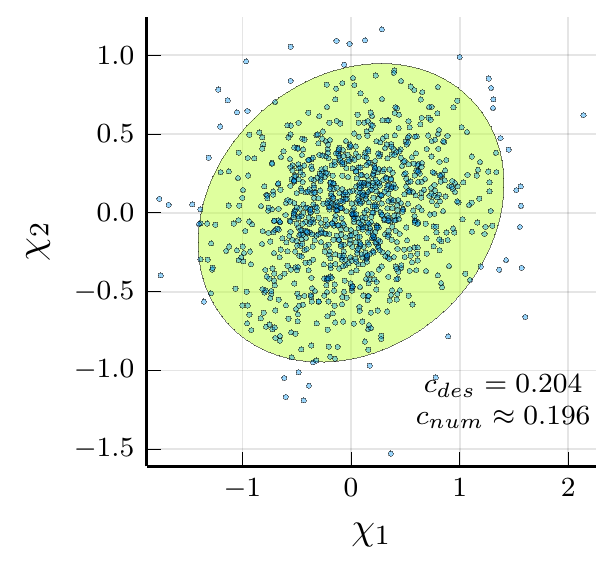}
  \end{subfigure}
  \medskip
  \begin{subfigure}[b]{.21\textwidth}
    \centering
    \includegraphics[width=\linewidth]{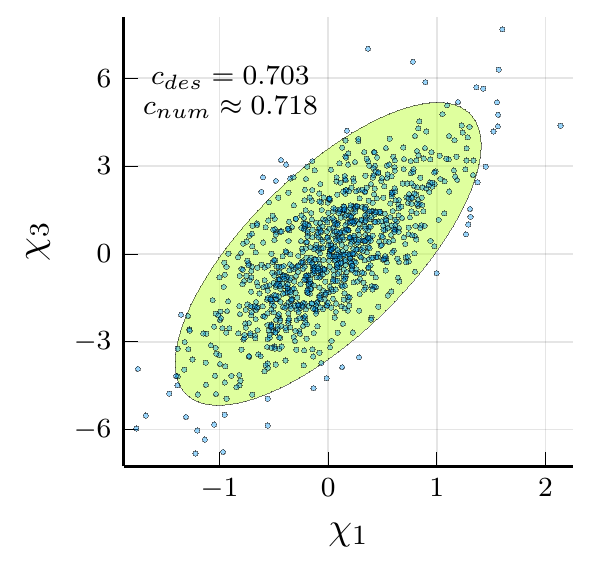}
  \end{subfigure}
  \hfill
  \begin{subfigure}[b]{.21\textwidth}
    \centering
    \includegraphics[width=\linewidth]{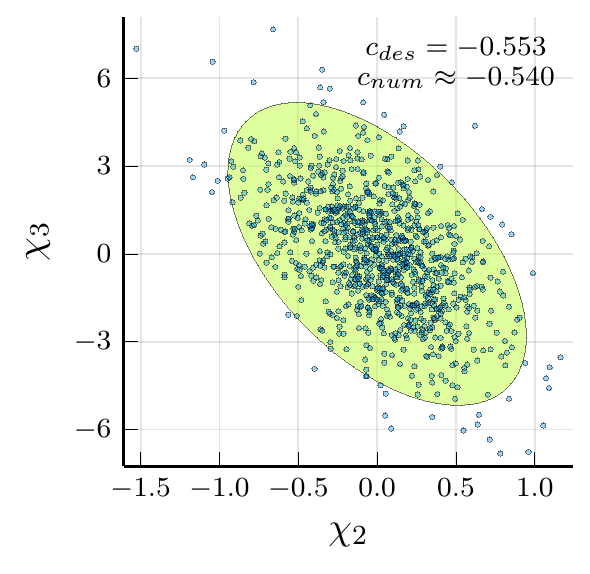}
  \end{subfigure}
\caption{Steady state solution for one thousand instances of model~\eqref{eq:main_system} in case the dominant eigenvalue of $A$ has algebraic multiplicity two but geometric multiplicity one. The computed correlations are close to the desired control values. Yellow ellipses are covariance ellipses.}
\label{fig: corr 2domJ}
\end{figure}

Figure~\ref{fig: corr 2domJ} illustrates Theorem~\ref{thm: noJor}. A solution to Problem~\ref{prob:main} was found by plugging expression~\eqref{eq: corr formulas Jor} into the Julia Optimization package {\tt Optim.jl}~\cite{mogensen2018optim}. In our algorithm, the third eigenvector $\bsv_3$ of matrix $A$, which is needed to define the matrix $U$ and its inverse, was defined as the unitary norm vector orthogonal to $\R\{\bsv_1,\bsv_2\}$.\footnote{Other choices are of course possible (e.g., keeping $\bsv_3$ fixed and equal to an arbitrarily chosen normalized vector).} The solution of our optimization procedure is the pair $\bsv_1,\bsv_2$ (and hence the orthonormal vector $\bsv_3$) and the slow eigenvalue $\lambda_1$. Given this solution, we assigned $\lambda_3=100\lambda_1$. The Hurwitz matrix $A$ was then constructed through the inverse of the change of base associated to the base $\bsv_1,\bsv_2,\bsv_3$. We then simulated one-thousand instances of model~\eqref{eq:main_system}. As shown in Figure~\ref{fig: corr 2domJ} numerically computed correlations are close to the associated desired control values. To verify that the solution set indeed lies on a two-dimensional manifold, we rerun our algorithm for the same desired correlation values as Figure~\ref{fig: corr 2domJ} but letting the optimization procedure converge to a new solution for five hundred initial conditions in a small neighborhood of the solution corresponding to Figure~\ref{fig: corr 2domJ}. The result of this experiment is reproduced in Figure~\ref{fig:nubes} (blue plots). The optimizer converged to disparate solutions. The dimension of the resulting set of solution was approximated via Principal Component Analysis using the package Julia {\tt MultivariateStats.jl}. As shown in Figure~\ref{fig:compara_vars}, only two dimensions consistently capture almost 100\% of the variance of the solution set, confirming Theorem~\ref{thm: Jor}.

\begin{figure}
    \centering
    \includegraphics[width=0.48\textwidth]{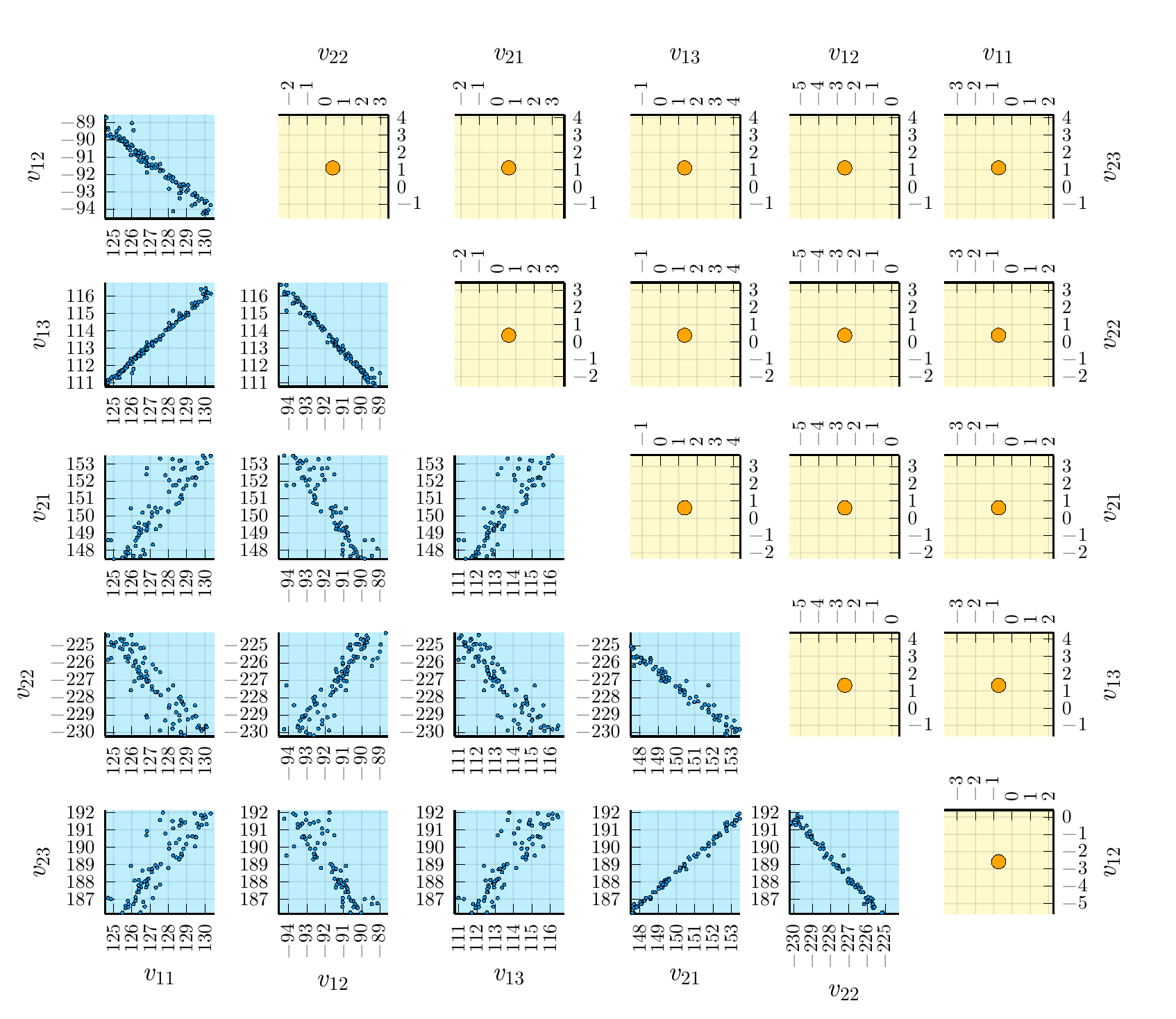}
    \caption{Projection of the components of the solution vectors $\bsv_1$ and $\bsv_2$ to Problem~\ref{prob:main} for the same cases as Figure~\ref{fig: corr 2domNJ} (yellow) and Figure~\ref{fig: corr 2domJ} (blue) when the optimization algorithm is initialized close to the solution of the respective figure.
    }
    \label{fig:nubes}
\end{figure}

\begin{figure}
    \centering
    \includegraphics[width=0.4\textwidth]{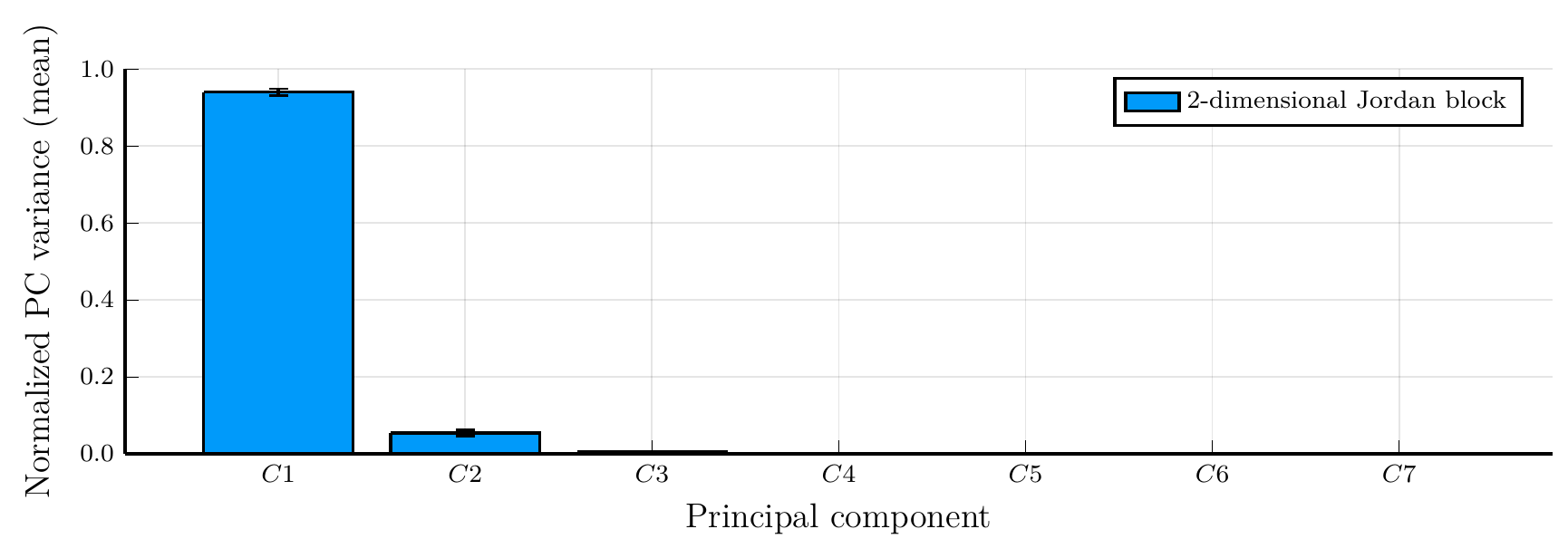}
    \caption{Principal component analysis (variance captured by the various components) of the solution set in Figure~\ref{fig:nubes}.}
    \label{fig:compara_vars}
\end{figure}

\section{Discussion}

Correlated variability is a fundamental property of biological neural systems but our understanding of it is still very poor. We introduced a control theoretical control problem that might lead to new insights about the functions and origins of correlated variability. The solution to this problem we started to sketch already revealed clear geometric connections between correlated variability and network structure. From a purely mathematical perspective, our work opens the question of which network structure implies two-dimensional dominant dynamics with geometric multiplicity one or two, in the same way as Metzler or eventually positive interconnection matrices lead to one-dominant dynamics.

\addtolength{\textheight}{-12cm}   

\bibliographystyle{IEEEtrans}
\bibliography{ACC2021}

\end{document}